\documentclass[smallextended]{svjour3}
\usepackage{amsmath,amssymb,mathrsfs}
\usepackage{color}
\usepackage{graphicx,subfigure}
\usepackage{indentfirst}
\usepackage{algorithm}
\usepackage{algorithmic}
\usepackage{xcolor}
\usepackage{marvosym}
\usepackage{natbib}
\newcommand{\envelope}{(\raisebox{-.5pt}{\scalebox{1.45}{\Letter}}\kern-.5pt)}

\newcommand{\be}{\begin{equation}}
\newcommand{\ee}{\end{equation}}

\spnewtheorem{subcase}{Subcase}[case]{\it}{}
\numberwithin{subcase}{case}
\begin{document}
\title{A Combination of Flow Shop Scheduling and the Shortest Path Problem\thanks{A preliminary version of this paper has appeared in the Proceedings of 19th Annual International Computing and Combinatorics Conference (COCOON'13), LNCS, vol. 7936, pp. 680--687.}}

\author{Kameng Nip \and Zhenbo Wang\and Fabrice Talla Nobibon \and Roel Leus 
}


\institute{Kameng Nip \at
              Department of Mathematical Sciences, Tsinghua University, Beijing, 100084, China\\
              \email{njm11@mails.tsinghua.edu.cn}           
           \and
           Zhenbo Wang \envelope \at
              Department of Mathematical Sciences, Tsinghua University, Beijing, 100084, China\\
              \email{zwang@math.tsinghua.edu.cn}
           \and
           Fabrice Talla Nobibon \at
                Postdoctoral researcher for Research Foundation--Flanders\\
                ORSTAT, Faculty of Economics and Business, KU Leuven, Belgium\\
              \email{Fabrice.TallaNobibon@kuleuven.be}
           \and
           Roel Leus \at
                ORSTAT, Faculty of Economics and Business, KU Leuven, Belgium\\
              \email{Roel.Leus@kuleuven.be}
}

\date{Received: date / Accepted: date}

\maketitle
\begin{abstract}
This paper studies a combinatorial optimization problem which is obtained by combining the flow shop scheduling problem and the shortest path problem. The objective of the obtained problem is to select a subset of jobs that constitutes a feasible solution to the shortest path problem, and to execute the selected jobs on the flow shop machines to minimize the makespan. We argue that this problem is NP-hard even if the number of machines is two, and is NP-hard in the strong sense for the general case. We propose an intuitive approximation algorithm for the case where the number of machines is an input, and an improved approximation algorithm for fixed number of machines.

\end{abstract}

\section{Introduction}
Combinatorial optimization is an active field in operations research and theoretical computer science. Historically, independent lines separately developed, such as machine scheduling, bin packing, travelling salesman problem, network flows, etc. With the rapid development of science and technology, manufacturing, service and management are often integrated, and decision-makers have to deal with systems involving several characteristics from more than one well-known combinatorial optimization problem. To the best of our knowledge, the combination of optimization problems has received only little attention in literature.

\citet{Bodlaender1994} studied parallel machine scheduling with incompatible jobs, in which two incompatible jobs cannot be processed by the same machine, and the objective is to minimize the makespan. This problem can be considered as a combination of parallel machine scheduling and the coloring problem.
\citet{WC12} studied a combination of parallel machine scheduling and the vertex cover problem. The goal is to select a subset of jobs that forms a vertex cover of a given graph and to execute these jobs on $m$ identical parallel machines to minimize the makespan. They proposed an $(3 - \frac{2}{m+1})$-approximation algorithm for this problem. \citet{WHH13} have investigated a generalization of the above problem that combines the uniformly related parallel machine scheduling problem and a generalized covering problem. They proposed several approximation algorithms and mentioned as future research other combinations of well-known combinatorial optimization problems. This is the core motivation for this work.

Let us consider the following scenario. We aim at building a railway between two specific cities. The railway needs to cross several adjacent cities, which is determined by a map (a graph). The processing time of manufacturing the rail track for each pair of cites varies between the pairs. Manufacturing a rail track between two cities in the graph is associated with a job. The decision-maker needs to make two main decisions: (1) choosing a path to connect the two cities, and (2) deciding the schedule of manufacturing the rail tracks on this path in the factory. In addition, the manufacturing of rail tracks follows several working stages, each stage must start after the completion of the preceding stages, and we assume that there is only one machine for each stage. We wish to accomplish the manufacturing as early as possible, i.e. minimize the last completion time; this is a standard flow shop scheduling problem. How should a decision-maker choose a feasible path such that the corresponding jobs can be manufactured as early as possible? This problem combines the structure of flow shop scheduling and the shortest path problem. Following the framework introduced by \citet{WHH13}, we can regard our problem as a combination of those two problems.

Finding a simple path between two vertices in a directed graph is a basic problem that can be polynomially solved \citep{AMO93}. Furthermore, if we want to find a path under a certain objective, various optimization problems come within our range of vision. The most famous one is the classic shortest path problem, which can be solved in polynomial time if the graph contains no negative cycle, and otherwise it is NP-hard \citep{AMO93}. Moreover, many optimization problems have a similar structure. For instance, the min-max shortest path problem \citep{KY97} studies a problem with multiple weights associated with each arc, and the objective is to find a directed path between two specific vertices such that the value of the maximum among all its total weights is minimized. The multi-objective shortest path problem \citep{Warburton87} also has multiple weights, but the objective is to find a Pareto optimal path between two specific vertices to satisfy some specific objective function.

Flow shop scheduling is one of the three basic models of multi-stage scheduling (the others are open shop scheduling and job shop scheduling). Flow shop scheduling with the objective of minimizing the makespan is usually denoted by $Fm||C_{max}$, where $m$ is the number of machines. In one of the earliest papers on scheduling problems, \citet{Joh54} showed that $F2||C_{max}$ can be solved in $O(n\log n)$ time, where $n$ is the number of jobs. On the other hand, \citet{GJS1976} proved that $Fm||C_{max}$ is strongly NP-hard for $m \geq 3$.

The contributions of this paper include: (1) a formal description of the considered problem, (2) the argument that the considered problem is NP-hard even if $m = 2$, and NP-hard in the strong sense if $m \geq 3$, and (3) several approximation algorithms.

The rest of the paper is organized as follows. In Section \ref{sec_pre}, we first give a formal definition of the problem stated above, then we briefly review flow shop scheduling and some shortest path problems, and introduce some related algorithms that will be used in the subsequent sections. In Section \ref{sec_com}, we study the computational complexity of the combined problem. Section \ref{sec_approx} provides several approximation algorithms for this problem. Some concluding remarks are provided in Section \ref{sec_end}.

\section{Preliminaries}\label{sec_pre}
\subsection{Problem Description}\label{sec_pd}
We first give a formal definition of our problem, which is a combination of the flow shop scheduling problem and the shortest path problem.

\begin{definition}\label{d_comb}
Given a directed graph $G = (V, A)$ with two distinguished vertices $s, t\in V$ and $m$ flow shop machines, each arc $a_j \in A$ corresponds with a job $J_j\in J$ with processing times $(p_{1j}$, $p_{2j}$, $\cdots$, $p_{mj})$ respectively. The $Fm|\mathrm{shortest}~\mathrm{path}|C_{max}$ problem is to find a $s-t$ directed path $P$, and to schedule the jobs of $J_P$ on the flow shop machines to yield the minimum makespan over all $P$, where $J_P$ denotes the set of jobs corresponding to the arcs in $P$.
\end{definition}

The considered problem is a combination of flow shop scheduling and the classic shortest path problem, mainly because the two optimization problems are special cases of this problem. For example, consider the following instances with $m = 2$. If there is a unique path from $s$ to $t$ in $G$, as shown in the left of Fig. \ref{figc}, our problem is the two-machine flow shop scheduling problem. If all the processing times on the second machine are zero, as shown in the right of Fig. \ref{figc}, then our problem is equivalent to the classic shortest path with respect to the processing times on the first machine. These examples illustrate that these two optimization problems are inherent in the considered problem.
\begin{figure}[ht]
  \includegraphics[width=4.5in]{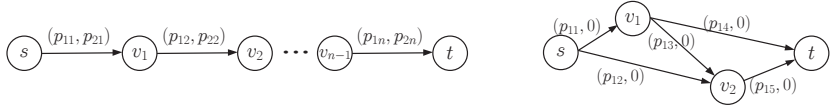}
  \caption{Special cases of our problem}\label{figc}
\end{figure}

In this paper, we will use the results of some optimization problems that have a similar structure with the classic shortest path problem. We introduce the following generalized shortest path problem.

\begin{definition}
Given a directed graph $G = (V, A, w^1, \cdots, w^K)$ and two distinguished vertices $s, t\in V$ with $|A| = n$. Each arc $a_j\in A, j = 1,\cdots,n$ is associated with $K$ weights $w^1_j, \cdots, w^K_j$, and we define the vector $w^k = (w^k_1, w^k_2, \cdots, w^k_n)$ for $k=1, 2, \cdots, K$. The goal of our shortest path problem $SP(G, s, t, f)$ is to find a $s-t$ directed path $P$ that minimizes $f(w^1, w^2, \cdots, w^K, x)$, in which $f$ is a given objective function and $x \in \{0, 1\}^n$ contains the decision variables such that $x_j = 1$ if and only if $a_j\in P$.
\label{d_sp}
\end{definition}

For ease of exposition, we use $SP$ instead of $SP(G,s,t,f)$ when there is no danger of confusion. Notice that $SP$ is a generalization of various shortest path problems. For instance, if we consider $K = 1$ and $f(w^1, x) = w^1\cdot x$, where $\cdot$ is the dot product, this problem is the classic shortest path problem. If $K = 2$ and $f(w^1, w^2, x) = \min\{w^1\cdot x: w^2\cdot x \leq W\}$, where $W$ is a given number, this problem is the shortest weight-constrained path problem \citep{GJ79}. If $f(w^1, w^2, \cdots, w^K, x) = \max\{w^1\cdot x, w^2\cdot x, \cdots, w^K\cdot x\}$, the problem is the min-max shortest path problem \citep{KY97}. In the following sections, we will analyze our combined problem by setting appropriate weights and objective function in $SP$.

\subsection{Algorithms for Flow Shop Scheduling Problems}\label{sec_f2}
First, we introduce some trivial bounds for flow shop scheduling. Denote by $C_{max}$ the makespan of an arbitrary flow shop schedule with job set $J$. A feasible shop schedule is called dense when any machine is idle if and only if there is no job that can be processed at that time on that machine. For arbitrary dense flow shop schedule, we have

\be
C_{max} \geq \max_{i\in\{1,\cdots,m\}}\left\{\sum_{J_j\in J}p_{ij}\right\},\label{eq_max}
\ee
and
\be
C_{max} \leq \sum_{J_j\in J}\sum^m_{i=1}p_{ij}.\label{eq_min}
\ee
For each job, we have
\be
C_{max} \geq \sum^m_{i=1}p_{ij},\qquad \forall J_j \in J.\label{eq_job}
\ee

In flow shop scheduling problems, a schedule is called a permutation schedule if all jobs are processed in the same order on each machine. \citet{Conway1971} proved that there always exists a permutation schedule which is optimal for $F2||C_{max}$ and $F3||C_{max}$. In a permutation schedule, the critical job and critical path are important concepts for the analysis of related algorithms.

Suppose we are given a job set $J$ with $n$ jobs. Let $\sigma = (\sigma(1), \cdots, \sigma(n))$ be a permutation of $(1,\cdots,n)$ for a three-machine (or two-machine) flow shop, and let $\{J_{\sigma(1)}, J_{\sigma(2)}, \cdots, J_{\sigma(n)}\}$ be the corresponding schedule. For simplicity of notation, we denote the permutation and the schedule by $(1, 2, \cdots, n)$ and $\{J_1, J_2, \cdots, J_n\}$ respectively. A directed graph is defined as follows. We define a vertex $(i,j)$ with an associated weight $p_{i,j}$ for each job $J_j$ and each machine $M_i$, for $i = 1, 2, 3$ (or $i = 1, 2$) and $j = 1,2,\cdots,n$. We include arcs leading from each vertex $(i, j)$ towards $(i + 1, j)$, and from $(i, j)$ towards $(i + 1, j + 1)$ for $j = 1, \cdots, n - 1$. The total weight of a maximum weight path from $(1,1)$ to $(3, n)$ (or $(2, n)$), which is called a critical path, is equal to the makespan of the corresponding permutation schedule. For the three-machine case, the critical jobs with respect to $\sigma$ are defined as the jobs $J_u$ and $J_v$ such that $(1,u)$, $(2,u)$, $(2,v)$ and $(3,v)$ appear in the critical path, i.e. the jobs $J_u$ and $J_v$ satisfy
\be
C_{max}  =  \sum^{u}_{j=1} p_{1j} + \sum^{v}_{j=u} p_{2j} + \sum^{n}_{j=v} p_{3j}.
\label{eq_criticaljob3}
\ee
For the two-machine case, the critical job with respect to $\sigma$ is defined as the job $J_{\nu}$ such that $(1,\nu)$ and $(2,\nu)$ appear in the critical path, i.e. the job $J_{\nu}$ satisfies
\be
C_{max}  = \sum^{\nu}_{j=1} p_{1j} + \sum^{n}_{j=\nu} p_{2j}.
\label{eq_criticaljob}
\ee

\citet{Joh54} proposed a sequencing rule for $F2||C_{max}$, which is one of the oldest results of the scheduling literature, and is commonly referred to Johnson's rule.
\begin{algorithm}[htb]
\caption{Johnson's rule}
\label{alg_jo}
\begin{algorithmic}[1]
\STATE Set $S_1 = \{J_j \in J | p_{1j}\leq p_{2j}\}$ and $S_2=\{J_j \in J| p_{1j}> p_{2j}\}$.
\STATE Process the jobs in $S_1$ first in a non-decreasing order of $p_{1j}$, and then schedule the jobs in $S_2$ in a non-increasing order of $p_{2j}$; ties may be broken arbitrarily.
\end{algorithmic}
\end{algorithm}

In Johnson's rule, jobs are scheduled as early as possible. This rule produces a permutation schedule, and Johnson showed that this schedule is optimal. Notice that this schedule is obtained in $O(n\log n)$ time.

For the general problem $Fm||C_{max}$, \citet{Gonzalez1978} first presented an $\lceil \frac{m}{2} \rceil$-approximation algorithm that runs in $O(mn \log{n})$ time by solving $\lceil \frac{m}{2} \rceil$ two-machine flow shop scheduling problems. \citet{Rock1982} proposed an alternative approach by reducing the original problem to an artificial two-machine flow shop problem; this approach is called machine aggregation heuristic. They obtained a permutation by solving the artificial problem in $O(mn + n\log n)$ time, and proved that it has the same performance guarantee of $\lceil \frac{m}{2} \rceil$. Based on the machine aggregation heuristic, \citet{Chen1996} proposed an algorithm for $F3||C_{max}$ with an improved performance guarantee of $\frac{5}{3}$. In the same paper, they also modified the Gonzalez and Sahni's algorithm if $m$ is odd, by partitioning the machines into $\frac{m - 3}{2}$ two-machine flow shop scheduling problems, and one three-machine flow shop scheduling problem which was solved by their $\frac{5}{3}$-approximation algorithm. The modified algorithm has the same performance ratio $\frac{m}{2}$ if $m$ is even, and an improved ratio $\frac{m}{2} + \frac{1}{6}$ if $m$ is odd. It is known that a PTAS exists for $Fm||C_{max}$ \citep{Hall1998}.

We refer to the aggregation heuristic of \citet{Rock1982} as the RS algorithm, and we will use it later to derive an algorithm for our combined problem. The RS algorithm can be described as follows for the three-machine case.
\begin{algorithm}[htb]
\caption{The RS algorithm for $F3||C_{max}$}
\label{alg_rs}
\begin{algorithmic}[1]
\STATE Construct an artificial two-machine flow shop scheduling problem with processing times $a_{j} = p_{1j} + p_{2j}$ on the first machine and $b_{j} = p_{2j} + p_{3j}$ on the second machine for $J_j \in J$. Implement Johnson's rule to obtain an optimal permutation $\sigma$ for the two-machine problem.
\STATE Assign the jobs on the three machines according to $\sigma$ as early as possible. Denote the makespan of this permutation schedule as $C_{max}$.
\RETURN $\sigma$, $C_{max}$.
\end{algorithmic}
\end{algorithm}

The running time of this algorithm is $O(n\log n)$, which is the same as Johnson's rule. Notice that the algorithm returns a permutation schedule, and hence the resulting makespan $C_{max}$ satisfies the equality (\ref{eq_criticaljob3}).

\subsection{Algorithms for Shortest Path Problems}\label{sec_sp}
In this paper, we will use the following two results of the shortest path problems. The first one is the well-known Dijkstra's algorithm, which solves the classic shortest path problem with nonnegative edge weights in $O(|V|^2)$ time \citep{DIJ59}. The second one is a FPTAS result for the min-max shortest path problem, which is presented by \citet{ABV06}. \citet{KY97} first proposed min-max criteria for several problems, including the shortest path problem. \citet{ABV06} studied the computational complexity and proposed several approximation schemes. The min-max shortest path problem with $K$ weights $w^1_j, \cdots, w^K_j$ associated with each arc $a_j$, is to find a path $P$ between two specific vertices that minimizes $\max_{k\in \{1, \cdots,K\}}\sum_{a_j\in P}w^k_j$. It was shown that this problem is NP-hard even for $K = 2$, and that a FPTAS exists if $K$ is a fixed number \citep{Warburton87,ABV06}. The algorithm of \citet{ABV06}, referred to the ABV algorithm in this paper, is based on dynamic programming and scaling techniques. The following result implies that the ABV algorithm is a FPTAS if $K$ is a constant.

\begin{theorem}[\citet{ABV06}]\label{th_minmax}
Given an arbitrary positive value $\epsilon>0$, in a given directed graph with $K$ nonnegative weights associated with each arc, a directed path $P$ between two specific vertices can be found by the ABV algorithm with the property
$$\max_{i\in \{1, 2, \cdots, K\}} \left\{\sum_{a_j\in P}w^i_j\right\} \leq (1+\epsilon) \max_{i\in \{1, 2, \cdots, K\}}\left\{\sum_{a_j\in P'}w^i_j\right\}$$ for any other path $P'$ between the two specific vertices, and the running time is $O(|A||V|^{K + 1}/\epsilon^K)$.
\end{theorem}

\section{Computational Complexity of $Fm|\mathrm{shortest}~\mathrm{path}|C_{max}$}\label{sec_com}
In this section, we study the computational complexity of our problem.
First, it is straightforward that $Fm|\mathrm{shortest}~\mathrm{path}|C_{max}$ is NP-hard in the strong sense if $m \geq 3$, as a consequence of the fact that $Fm||C_{max}$ is a special case of our problem.

On the other hand, although $F2||C_{max}$ and the classical shortest path are polynomially solvable, we argue that $F2|\mathrm{shortest}~\mathrm{path}|C_{max}$ is NP-hard. We prove this result by using a reduction from the NP-complete problem \textsc{partition} \citep{GJ79}. Our proof is similar to the well-known NP-hardness proof for the shortest weight-constrained path problem \citep{Batagelj2000}.
\begin{theorem}
$Fm|\mathrm{shortest}~\mathrm{path}|C_{max}$ is NP-hard even if $m = 2$, and is NP-hard in the strong sense for $m \geq 3$. \label{th_npc}
\end{theorem}
\begin{figure}[ht]
  \centering
  \includegraphics[width=4.5in]{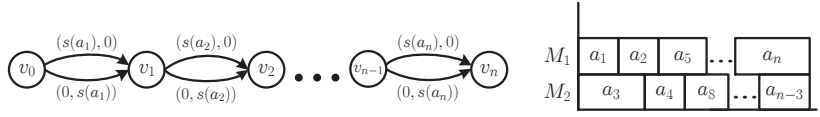}\\
  \caption{The reduction from \textsc{partition} to $F2|\mathrm{shortest}~\mathrm{path}|C_{max}$}\label{figc2}
\end{figure}
\begin{proof}We only need to prove the first part. It is easy to see that the decision version of $F2|\mathrm{shortest}~\mathrm{path}|C_{max}$ belongs to NP. Consider an arbitrary instance of \textsc{partition} with $S = \{a_1, \cdots, a_n\}$ with size $s(a_k)\in \mathbb{Z}^+$ for each $k$, and let $C = \sum_{a\in S}s(a)/2$. We now construct the directed graph $G = (V, A, W)$ and the corresponding jobs. The graph has $n + 1$ vertices $v_0, v_1, \cdots, v_n$, each pair of $(v_k, v_{k+1}),k=0,1,\cdots,n-1,$ is joined by two parallel arcs (jobs) with processing times $(s(a_{k+1}), 0)$ and $(0, s(a_{k+1}))$ respectively, both leading from vertex $v_k$ towards $v_{k+1}$ (see the left of Fig. \ref{figc2}). We wish to find the jobs corresponding to a path from $v_0$ to $v_{n+1}$. It is not difficult to check that there is a feasible schedule with makespan not more than $C$ if and only if there is a partition of set $S$ (see the right of Fig. \ref{figc2}). Therefore, the decision version of $F2|\mathrm{shortest}~\mathrm{path}|C_{max}$ is $NP$-complete. $\Box$
\end{proof}

\section{Approximation Algorithms}\label{sec_approx}
\subsection{An intuitive Algorithm}\label{sec_alg_nat}
To start off, we propose an intuitive algorithm for $Fm|\mathrm{shortest}~\mathrm{path}|C_{max}$. The main idea of this algorithm is to set $K = 1$ and $f = w^1 \cdot x$ in $SP$, i.e. to find a classical shortest path with one specific set of weights. An intuitive setting is $w^1_j = \sum^m_{i = 1}p_{ij}$ for each arc. We find the shortest path with respect to $w^1$ by Dijkstra's algorithm, and then schedule the corresponding jobs on the flow shop machines. We refer to this algorithm as the FD algorithm. The subsequent analysis will show that the performance ratio of the FD algorithm remains the same for an arbitrarily selected flow shop scheduling algorithm that provides a dense schedule, regardless of the performance ratio of the algorithm.

\begin{algorithm}[htb]
\caption{The FD algorithm}
\label{alg_1}
\begin{algorithmic}[1]
\STATE Find the shortest path in $G$ with weights $w^1_j := \sum^m_{i = 1}p_{ij}$ by Dijkstra's algorithm. For the returned path $P$, construct the job set $J_P$.
\STATE Obtain a dense schedule of the jobs of $J_P$ by an arbitrary flow shop scheduling algorithm. Let $\sigma$  be the returned job schedule and $C_{max}$ the returned makespan, and denote the job set $J_P$ by $S$.
\RETURN $S$, $\sigma$ \AND $C_{max}$
\end{algorithmic}
\end{algorithm}

It is straightforward that the total running time of the FD algorithm is $O(|V|^2 + T(m, n))$, where $T(m, n)$ is the running time of the flow shop scheduling algorithm. Therefore, suppose the flow shop scheduling algorithm we used is polynomial time, then the FD algorithm is polynomial time even if $m$ is an input of the instance. Before we analyze the performance of this algorithm, we first introduce some notations. Let $J^*$ be the set of jobs in an optimal solution, and $C^*_{max}$ be the corresponding makespan, and let $S$ and $C_{max}$ be those returned by the FD algorithm respcetively.
\begin{theorem}
The FD algorithm is $m$-approximate, and this bound is tight.
\end{theorem}

\begin{proof}
By the lower bound (\ref{eq_max}) introduced in Section \ref{sec_f2}, we have
\begin{equation}\label{eq:pf_hd_1}
mC^*_{max}  \geq  \sum_{J_j\in J^*}\sum^{m}_{i = 1}p_{ij}. \end{equation}

Since the returned path is a shortest path with respect to $w^1$, by (\ref{eq_min}) we have
\be\label{eq:pf_hd_2}
C_{max}\leq \sum_{J_j\in S}\sum^m_{i=1}p_{ij} = \sum_{J_j\in S} w^1_j \leq \sum_{J_j\in J^*} w^1_j = \sum_{J_j\in J^*}\sum^m_{i=1}p_{ij}.
\ee

By combining (\ref{eq:pf_hd_1}) with (\ref{eq:pf_hd_2}), it follows that $C_{max} \leq mC^{*}_{max}.$

\begin{figure}[ht]
  \centering
  \includegraphics[width=4in]{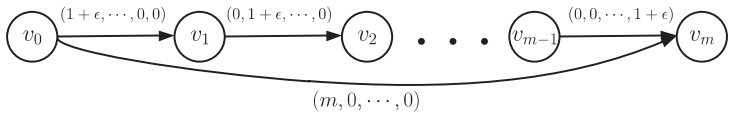}\\
  \caption{Tight example for the FD algorithm}\label{figfdtight}
\end{figure}

Consider the instance shown in Fig. \ref{figfdtight}. We wish to find a path from vertex $v_0$ to $v_m$. The makespan returned by the FD algorithm is $C_{max} = m$ with the arc $(v_0, v_m)$, whereas the makespan of an optimal schedule is $C_{max}^* = 1 + \epsilon$ with the other arcs. Notice that there is only one job in the returned solution, hence the returned makespan remains $m$ regardless of the algorithm used for the flow shop scheduling. The bound is tight because $\frac{C_{max}}{C_{max}^*}\rightarrow m$ when $\epsilon \rightarrow 0$.  $\Box$
\end{proof}

\subsection{An Improved Algorithm for Fixed $m$}\label{sec_alg_im}
In this subsection, we assume that $m$, the number of flow shop machines, is a constant. Instead of finding an optimal shortest path from $s$ to $t$ with respect to specific weights, we implement the ABV algorithm mentioned in Section \ref{sec_sp}, which will return a $(1+\epsilon)$-approximated solution for the min-max shortest path problem. In other words, we will set $K=m$ and use the objective function $f = \max\{w^1\cdot x, w^2\cdot x, \cdots, w^K\cdot x\}$ in $SP$, where the wights $w^1, w^2, \cdots, w^K$ will be decided later.

Inspired by the work of \citet{Gonzalez1978} and \citet{Chen1996}, we proceed as follows: after obtaining a feasible path by the ABV algorithm, we schedule the corresponding jobs by partitioning the $m$ machines into several groups. Denote the machine as $M_i$, $i = 1, \cdots, m$ (indexed following the routing of the flow shop). More specifically, we partition the $m$ machines into $m_3$ groups of three consecutive machines in the routing $M_{3i-2}$, $M_{3i-1}$, $M_{3i}$ ($i=1,\cdots, m_3$), $m_2$ groups of two consecutive machines in the routing $M_{3m_3 + 2i-1}$, $ M_{3m_3 + 2i}$ ($i=1,\cdots, m_2$), and $m_1$ individual machines $M_{3m_3 + 2m_2 + i}$ ($i=1,\cdots, m_1$), in which the value of $m_1$, $m_2$, $m_3$ will be derived later. For the three-machine subproblems on $M_{3i-2}$, $M_{3i-1}$ and $M_{3i}$ ($i=1,\cdots, m_3$), we implement the RS algorithm to obtain the permutations. For the two-machine subproblems on $M_{3m_3 + 2i-1}$ and $M_{3m_3 + 2i}$ ($i=1,\cdots, m_2$), we implement Johnson's rule to obtain the permutations. The permutations for the single-machine subproblems are arbitrary. Then we form a schedule for the original $m$-machine problem, in which the sequences of jobs on machines $M_i$ are the permutations obtained above, and are executed as early as possible. Notice the property that an optimal schedule is always a permutation schedule only stands for $F2||C_{max}$ and $F3||C_{max}$, and the performance guarantee relies on the properties of critical jobs as we will see in the subsequent analysis. The reason why we partition the $m$ machines in this particular fashion is related to this fact, as will be explained below.

The main idea of our algorithm is described as follows. We initially set the weights $(w^1_j, w^2_j, \cdots, w^m_j)= (p_{1j}, p_{2j}, \cdots, p_{mj})$. The algorithm iteratively runs the ABV algorithm and the above partition scheduling algorithm (the values of $m_1, m_2, m_3$  will be decided later) by adopting the following revision policy: in a current schedule, if there exists a job whose weight is large enough with respect to the current makespan, we will modify the weights of arcs corresponding to large jobs to $(M, M, \cdots, M)$, where $M$ is a sufficient large number, and then mark these jobs. The algorithm terminates if no such job exists. Another termination condition is that a marked job appears in a current schedule. We return the schedule with minimum makespan among all current schedules as the solution of the algorithm. We refer to this algorithm as the PAR algorithm. Notice that the weights of arcs may vary in each iteration, whereas the processing times of jobs remain the same throughout this algorithm.

Before we formally state the PAR algorithm, we first provide more details about the parameter choices. For $m = 2$ and $m = 3$, by following the subsequent analysis of the performance of this algorithm, one can verify that the best possible performance ratio is $\frac{3}{2}$ and $2$ respectively. An intuitive argument is that the best possible performance ratio for the general case of the PAR algorithm is $\rho= m_1+ \frac{3}{2}m_2 + 2m_3$. For a given $m$, as $m_1, m_2, m_3$ are nonnegative integers, our task is to minimize $m_1+ \frac{3}{2}m_2 + 2m_3$ such that $m_1 + 2m_2 + 3m_3 = m$. A simple calculation yields the following result:

\begin{eqnarray}\label{alg_pfar_p1}
(m_1, m_2, m_3)=\left\{ \begin{array}{ll}
           (0, 0, \frac{m}{3}) & \mathrm{~if~} m = 0\pmod{3} ,\\
           (1, 0, \frac{m-1}{3}) & \mathrm{~if~} m = 1 \pmod{3},\\
           (0, 1, \frac{m-2}{3}) & \mathrm{~if~} m = 2 \pmod{3},
          \end{array}\right.
\end{eqnarray}
and
\begin{eqnarray}\label{alg_pfar_p2}
\rho=\left\{ \begin{array}{ll}
           \frac{2m}{3} & \mathrm{~if~} m = 0 \pmod{3},\\
           \frac{2m + 1}{3} & \mathrm{~if~} m = 1 \pmod{3},\\
           \frac{4m + 1}{6} & \mathrm{~if~} m = 2 \pmod{3}.
          \end{array}\right.
\end{eqnarray}

In other words, the best way is to partition the machines in such a way that we have a maximum number of three-machine subsets. The pseudocode of the PAR algorithm is described by Algorithm 4.

\begin{algorithm}[htb]
\caption{The PAR algorithm}
\label{alg_pfmar}
\begin{algorithmic}[1]
\STATE Derive $(m_1,m_2,m_3)$ and $\rho$ using (\ref{alg_pfar_p1}) and (\ref{alg_pfar_p2}).

\STATE Initially, $(w^1_j, w^2_j, \cdots, w^m_j) := (p_{1j}, p_{2j}, \cdots, p_{mj})$ for each arc $a_j\in A$ corresponding to $J_j\in J$.

\STATE Given $\epsilon >0$, implement the ABV algorithm to obtain a path $P$ for $SP$, and construct the corresponding job set as $J_P$.

\STATE Partition the $m$ machines: $m_3$ three-machines subsets ($M_{3i-2}, M_{3i-1}, M_{3i}$, $i=1,\cdots, m_3$); one two-machine subsets ($M_{m-1}$ and $M_{m}$) if $m_2 = 1$; one single-machine subset ($M_{m}$) if $m_1 = 1$. \label{alg_pfar_lseq}

\STATE Run RS algorithm to obtain the permutations for these three-machine flow shops, and Johnson's rule to obtain the permutation for the two-machine flow shop. Let the sequence of the single-machine problem be arbitrary.

\STATE For the original problem, schedule the jobs of $J_P$ according to those permutations on each machine as early as possible. Denote the returned makespan as $C'_{max}$, and the job schedule as $\sigma'$. \label{alg_pfar_lseq2}

\STATE $S : = J_P$, $\sigma: = \sigma'$, $C_{max}:=C'_{max}$, $D:=\emptyset$, $M : = (1+\epsilon)\sum_{J_j\in J}\sum_{i = 1}^{m}p_{ij} + 1$.

\WHILE{$J_P \cap D = \emptyset $ \AND there exists a job $J_j$ in $J_P$ such that $\sum_{i=1}^mp_{ij} > \frac{C'_{max}}{\rho}$}
    \FOR {all jobs with $\sum_{i=1}^mp_{ij} > \frac{C'_{max}}{\rho}$ in $J\backslash D$}
        \STATE $(w^1_j, w^2_j, \cdots, w^m_j) := (M, M, \cdots, M)$, $D := D\cup \{J_j\}$.
    \ENDFOR
    \STATE Implement the ABV algorithm to obtain a path $P$ to $SP$, and construct the corresponding job set as $J_P$.

    \STATE Schedule the jobs of $J_P$ by the rule described in lines \ref{alg_pfar_lseq} -- \ref{alg_pfar_lseq2}.

    \IF{$C'_{max} < C_{max}$}
        \STATE $S : = J_P$, $\sigma: = \sigma'$, $C_{max}:=C'_{max}$.
    \ENDIF
\ENDWHILE
\RETURN $S$, $\sigma$ \AND $C_{max}$.
\end{algorithmic}
\end{algorithm}

It is easy to see that the PAR algorithm will return a feasible solution of $Fm|\mathrm{shortest}~\mathrm{path}|C_{max}$. We now discuss its computational complexity. Let the total number of jobs be $|A| = n$. Notice that the weights of arcs can be revised at most $n$ times. It is straightforward that the total running time of the PAR algorithm is $O(n^2|V|^{m + 1}/\epsilon^m + mn^2\log n)$, since there are at most $n$ iterations, in which the running time of the ABV algorithm is $O(n|V|^{m + 1}/\epsilon^m)$ and scheduling takes $O(mn\log n)$ time. If $m$ and $\epsilon$ are fixed numbers, then the PAR algorithm is a polynomial time algorithm.

Let $J^*$ be the set of jobs in an optimal solution, and $C^*_{max}$ the corresponding makespan, and let $S$ and $C_{max}$ be those returned by the PAR algorithm respectively. The following theorem shows the performance of the PAR algorithm.
\begin{theorem}
Given $\epsilon >0$, the worst-case ratio of the PAR algorithm for $Fm|\mathrm{shortest}~\mathrm{path}|C_{max}$ is
\begin{eqnarray}
(1 + \epsilon)\rho =\left\{ \begin{array}{ll}
           (1+\epsilon)\frac{2m}{3}& \mathrm{~if~} m = 0 \pmod{3},\\
           (1+\epsilon)\frac{2m + 1}{3} & \mathrm{~if~} m = 1 \pmod{3},\\
           (1+\epsilon)\frac{4m + 1}{6}& \mathrm{~if~} m = 2 \pmod{3}.
          \end{array}\right.
\end{eqnarray}
\end{theorem}

\setcounter{case}{0}
\begin{proof}
We will distinguish two different cases: $J^*\cap D \neq \emptyset$ and $J^*\cap D = \emptyset$.

\begin{case} $J^*\cap D \neq \emptyset$

In this case, there is at least one job in the optimal solution, say $J_j$, such that $C'_{max} < \rho\sum_{i = 1}^{m}p_{ij}$ holds for a current schedule with makespan $C'_{max}$ during the execution. Notice that the schedule returned by the PAR algorithm is the schedule with minimum makespan among all current schedules, and we have $C_{max} \leq C'_{max}$. It follows from (\ref{eq_job}) that
\be
\begin{split}
C_{max}  \leq C'_{max} &< \rho\sum_{i = 1}^{m}p_{ij} \leq \rho C^*_{max}.
\end{split}
\ee

\end{case}
\begin{case}
$J^*\cap D = \emptyset$

Consider the last current schedule during the execution of the algorithm. Denote the corresponding job set and the makespan as $J'$ and $C'_{max}$ respectively.

In this case, we first argue that $J' \cap D = \emptyset$. Suppose that this is not the case, i.e. $J' \cap D \neq \emptyset$. Since $J^*\cap D = \emptyset$, we know the weights of arcs corresponding to the jobs in $J^*$ have not been revised. Hence we have $ (1+\epsilon)\max_{i\in\{1,\cdots, m\}} \left\{\sum_{J_j\in J^*}w^i_j\right\} < M$. Moreover, by the assumption $J' \cap D \neq \emptyset$, we have $\max_{i\in\{1,\cdots, m\}} \left\{\sum_{J_j\in J'}w^i_j \right\} \geq M$. By Theorem \ref{th_minmax}, the solution returned by the ABV algorithm satisfies
\begin{equation*}
M \leq \max_{i\in\{1,\cdots, m\}} \left\{\sum_{J_j\in J'}w^i_j \right\} \leq (1+\epsilon)\max_{i\in\{1,\cdots, m\}} \left\{\sum_{J_j\in J^*}w^i_j\right\} < M,
\end{equation*} which leads to a contradiction.

Remember that in the PAR algorithm, the machines are divided into three parts, namely three-machines subsets together with at most one two-machine subset or a single machine. We solve these subproblems by the RS algorithm, Johnson's rule and an arbitrary algorithm respectively. It is clear that the sum of the makespans of those schedules is an upper bound for $C'_{max}$. Denote $C^2_{max}$ and $J^2_{\nu}$ as the makespan and the critical job of the two-machine subproblem returned by Johnson's rule, and let the corresponding machines be $M_{i_2}, M_{i_2 + 1}$. Denote $C^3_{max}$ and $J^3_u$, $J^3_v$ as the makespan and the critical jobs returned by the RS algorithm for the three-machine subproblems with largest makespan, and let the machines be $M_{i_3}, M_{i_3 + 1}, M_{i_3 + 2}$. Denote the single machine as $M_{i_1}$, on which the total processing time is $\sum_{J_j \in J'}p_{i_1j}$.

For the two-machine case, suppose that $p_{i_2,\nu} \geq p_{i_2+1,\nu}$. Noticing that $p_{i_2,j}\geq p_{i_2+1,j}$ for the job scheduled after $J_{\nu}$ in the schedule returned by Johnson's rule and form (\ref{eq_criticaljob}), it follows that
\be
C^2_{max} \leq \sum_{J_j\in J'} p_{i_2,j} +  p_{i_2+1,\nu} \leq \sum_{J_j\in J'} p_{i_2,j} +  \frac{1}{2}(p_{i_2,\nu} + p_{i_2+1,\nu}). \label{eq_crit2}
\ee

For the three-machine case, we study two subcases corresponding with $u = v$ and $u < v$ for the critical jobs.

\begin{subcase}
$u = v$.

Consider the schedule with respect to $C^3_{max}$. We can rewrite (\ref{eq_criticaljob3}) as
\be
C^3_{max} = \sum^{u-1}_{j=1}p_{i_3,j} + p_{i_3,u} + p_{i_3 + 1, u} + p_{i_3 + 2,u} + \sum^{n}_{j=u+1}p_{i_3+2, j}. \label{alg_rf3ar_1}
\ee

Suppose that the processing times of the critical jobs of the three-machine subproblem satisfy $p_{i_3,u} \geq p_{i_3+2,u}$, thus we have $p_{i_3u} + p_{i_3+1,u} \geq p_{i_3+1,u} + p_{i_3+2,u}$, i.e. $a_{u} \geq b_{u}$ for the artificial two-machine flow shop in the RS algorithm. Since the RS algorithm schedules the jobs by Johnson' rule, thus we have $a_{j} \geq b_{j}$ for the jobs scheduled after $J_{u}$, i.e. $p_{i_3,j} \geq p_{i_3+2,j}$. From (\ref{alg_rf3ar_1}), we have
\be
C^3_{max} \leq \sum_{J_j\in J'}p_{i_3,j} + p_{i_3,u} + p_{i_3+1,u} + p_{i_3+2,u}. \label{alg_rf3ar_2}
\ee

Since $J' \cap D = \emptyset$, we know the weights of arcs corresponding to the jobs in the last current schedule have not been revised, and $\sum_{i = 1}^{m}p_{ij} \leq \frac{C'_{max}}{\rho}$ for each job $J_j\in J'$, since otherwise the algorithm will continue. Since $J^* \cap D = \emptyset$, the weights of arcs corresponding to the jobs in this optimal schedule have not been revised. Thus, it follows from (\ref{eq_max}), (\ref{eq_job}), Theorem~\ref{th_minmax}, (\ref{eq_crit2}), (\ref{alg_rf3ar_2}) and the fact that the schedule returned by the PAR algorithm is the schedule with minimum makespan among all current schedules, that
\begin{equation*}
\begin{split}
C_{max} \leq C'_{max} & \leq m_3 C^3_{max} + m_2 C^2_{max} + m_1 \sum_{J_j \in J'}p_{i_1,j}\\
& \leq m_3\left(\sum_{J_j \in J'}p_{i_3,j} + p_{i_3,u} + p_{i_3 + 1,u} + p_{i_3 + 2,u}\right) \\ & \quad + m_2\left(\sum_{J_j \in J'}p_{i_2,j} + \frac{1}{2}(p_{i_2,\nu} + p_{i_2+1,\nu})\right) + m_1 \sum_{J_j \in J'}p_{i_1,j}\\
& \leq m_3\left((1+\epsilon)\max_{i\in \{1, \cdots, m\}} \left\{\sum_{J_j\in J^*}p_{ij}\right\} + \frac{C'_{max}}{\rho}\right) \\& \quad + m_2\left((1+\epsilon)\max_{i\in \{1, \cdots, m\}} \left\{\sum_{J_j\in J^*}p_{ij}\right\} + \frac{C'_{max}}{2\rho}C'_{max}\right) \\
& \quad + m_1 (1+\epsilon)\max_{i\in \{1, \cdots, m\}} \left\{\sum_{J_j\in J^*}p_{ij}\right\}\\
& \leq (m_1 + m_2 + m_3)(1+\epsilon)C^*_{max}+ \left(\frac{m_2}{2\rho} + \frac{m_3}{\rho}\right) C'_{max}\\
\end{split}
\end{equation*}

Substituting (\ref{alg_pfar_p1}) and (\ref{alg_pfar_p2}) into the $m_1$, $m_2$, $m_3$ and $\rho$, by a simple calculation, we arrive at
\be
C_{max}\leq C'_{max} \leq (1 + \epsilon)\rho C^*_{max}.
\ee
\end{subcase}
\begin{subcase}
$u < v$.

We also assume that $p_{i_3,u} \geq p_{i_3+2,u}$, and $p_{i_2,\nu} \geq p_{i_2 + 1,\nu}$, an argument similar to the previous case shows that the jobs scheduled after $J_{v}$ satisfies $p_{i_3,j} \geq p_{i_3+2,j}$. Since $u < v$, it follows from (\ref{eq_criticaljob3}) that
\be
C^3_{max} \leq \sum_{J_j \in J'} p_{i_3,j} + \sum^{v}_{j = u} p_{i_3 + 1, j} \leq \sum_{J_j \in J'} p_{i_3,j} + \sum_{J_j \in J'} p_{i_3+1, j}. \label{eq_rf3ar_2}
\ee

Similarly, it is not difficult to show that
\begin{equation*}
\begin{split}
C_{max} \leq C'_{max} & \leq m_3 C^3_{max} + m_2 C^2_{max} + m_1 \sum_{J_j \in J'}p_{i_1,j}\\
& \leq m_3\left(\sum_{J_j \in J'} p_{i_3,j} + \sum_{J_j \in J'} p_{i_3+1, j}\right) \\
&\quad + m_2\left(\sum_{J_j \in J'}p_{i_2,j} + \frac{1}{2}(p_{i_2,\nu} + p_{i_2+1,\nu})\right) + m_1 \sum_{J_j \in J'}p_{i_1,j}\\
& \leq (m_1 + m_2 + 2m_3)(1+\epsilon)C^*_{max}+ \frac{m_2}{2\rho} C'_{max}\\
\end{split}
\end{equation*}

Substituting (\ref{alg_pfar_p1}) and (\ref{alg_pfar_p2}) into $m_1$, $m_2$, $m_3$ and $\rho$, by a simple calculation, we obtain
\be
C_{max}\leq C'_{max} \leq (1 + \epsilon)\rho C^*_{max}.
\ee
\end{subcase}
For the cases where the last current schedule has critical jobs satisfying $p_{i_2,\nu} < p_{i_2 + 1,\nu}$ or $p_{i_3,u} < p_{i_3+2,u}$, analogous arguments would yield the same result.
\end{case}

Now we show that the performance ratio of the PAR algorithm cannot less than $\rho$. First, we propose two instances for $m = 2$ and $m = 3$.
\begin{figure}[ht]
  \centering
  \includegraphics[width=3in]{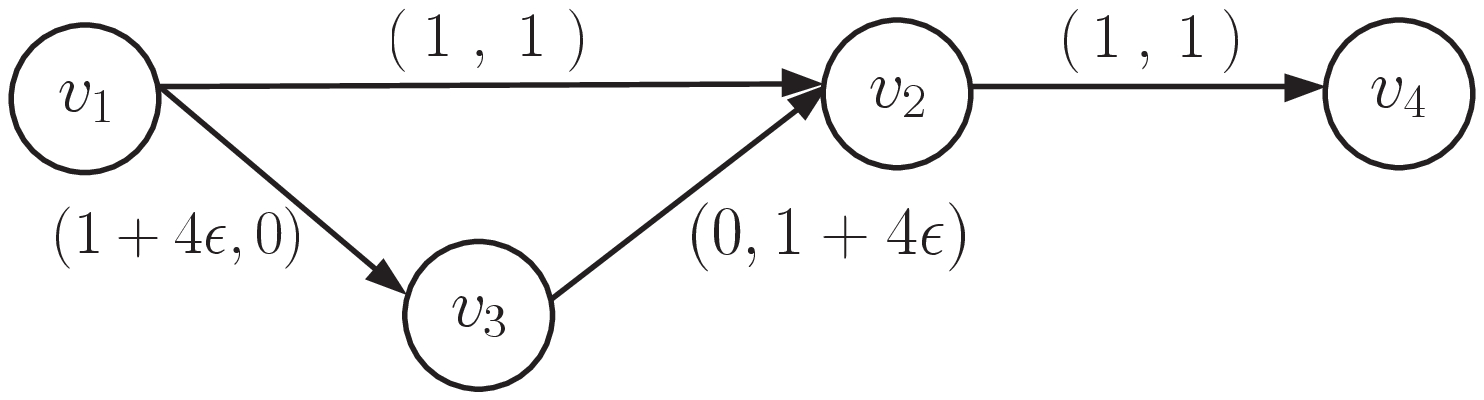}\\
  \includegraphics[width=4in]{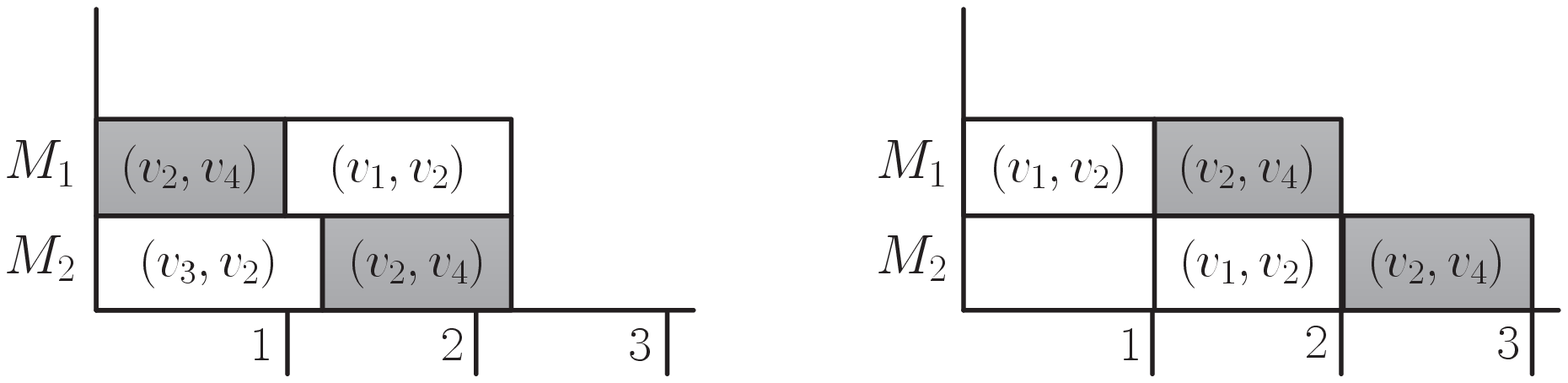}\\
  \caption{Example for $m = 2$}\label{fig2tight}
\end{figure}

If $m = 2$, the performance ratio of the PAR algorithm is $\frac{3}{2}(1+\epsilon)$. Consider the following instance shown in Fig. \ref{fig2tight}. We wish to find a path from $v_1$ to $v_4$. Notice that the ABV algorithm returns the path with arcs $(v_1, v_2)$ and $(v_2, v_4)$, and the corresponding makespan $C'_{max}$ by Johnson's rule is $3$. All the corresponding jobs satisfy $p_{1j} + p_{2j} = 2 \leq \frac{2}{3}C'_{max}$, and thus the algorithm terminates. Therefore, the makespan of the returned schedule by the PAR algorithm is $C_{max} = 3$ (see the right schedule of Fig. \ref{fig2tight}). On the other hand, the optimal makespan is $C_{max}^* = 2+4\epsilon$ with arcs $(v_1, v_3)$, $(v_3, v_2)$ and $(v_2, v_4)$ (see the left schedule of Fig. \ref{fig2tight}). The worst case ratio of the PAR algorithm cannot be less than $\frac{3}{2}$ as $\frac{C_{max}}{C_{max}^*}\rightarrow 3/2$ when $\epsilon \rightarrow 0$ for this instance.

\begin{figure}[ht]
  \centering
  \includegraphics[width=3.5in]{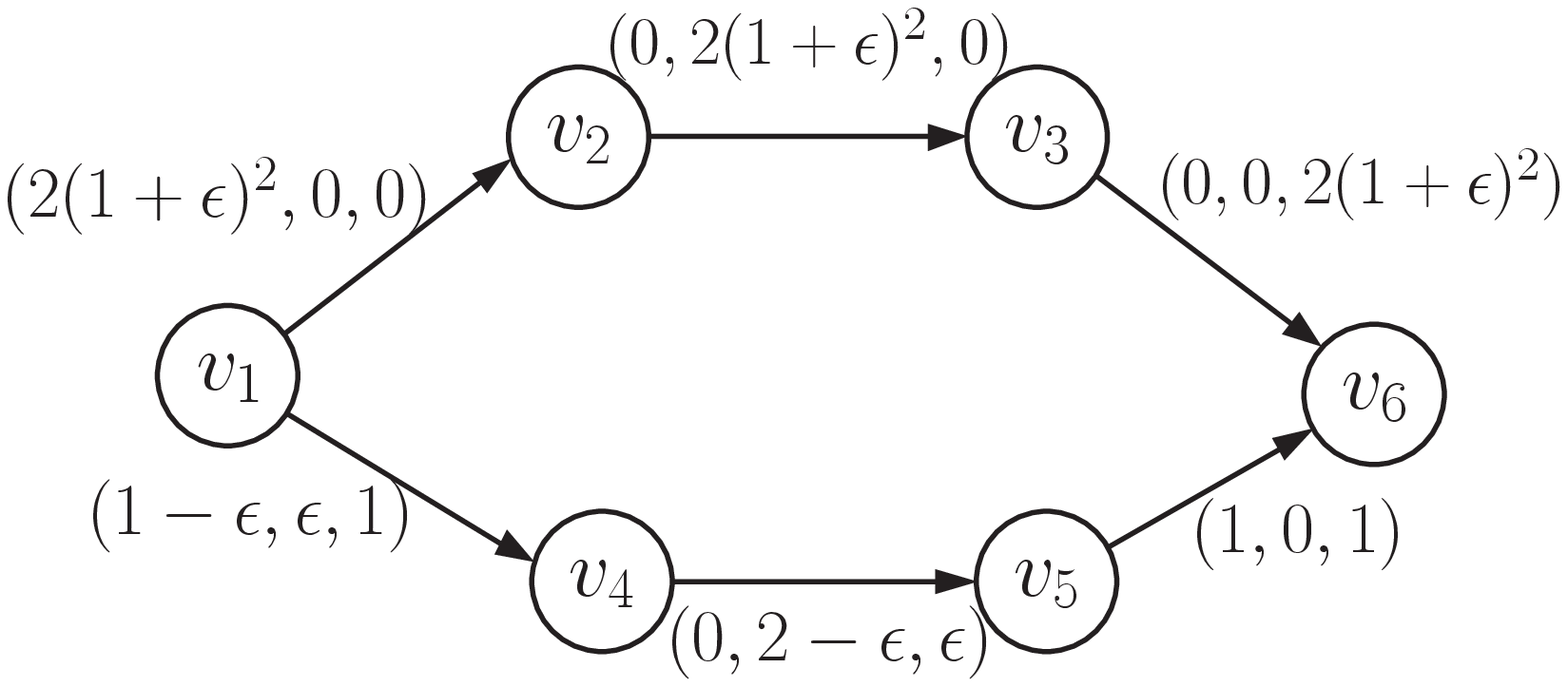}\\
  \includegraphics[width=3.5in]{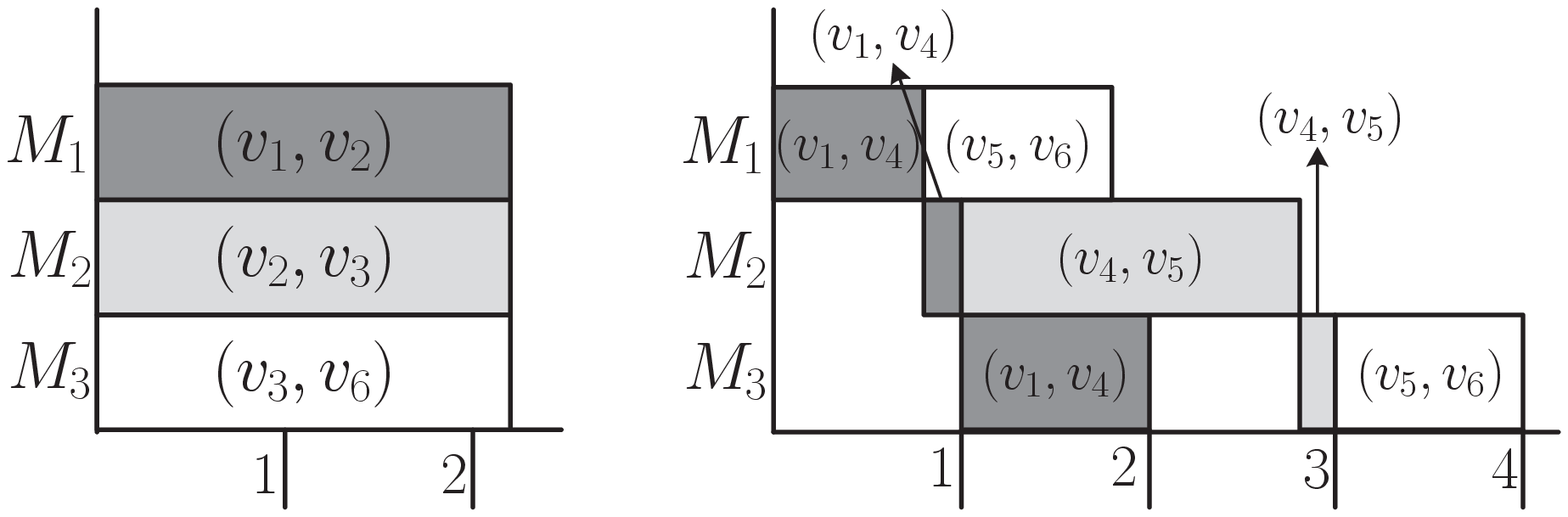}\\
  \caption{Example for $m = 3$}\label{fig3tight}
\end{figure}

For the case where $m = 3$, the performance ratio of the PAR algorithm is $2(1+\epsilon)$. Consider the instance shown in Fig. \ref{fig3tight}. We wish to find a path from vertex $v_1$ to $v_6$. Notice that the ABV algorithm returns the path with arcs $(v_1, v_4) \rightarrow (v_4, v_5) \rightarrow (v_5, v_6)$. The makespan of the schedule returned by the RS algorithm is $C'_{max} = 4$. All the corresponding jobs satisfy $p_{1j} + p_{2j} = 2 \leq \frac{1}{2}C'_{max}$, and thus the algorithm terminates. Therefore, the makespan of the schedule returned by the PAR algorithm is $C_{max} = 4$ (see the right schedule of Fig. \ref{fig3tight}). On the other hand, the makespan of an optimal job schedule is $C_{max}^* = 2(1 + \epsilon)^2$, by selecting the arcs $(v_1, v_2)$, $(v_2, v_3)$, and $(v_3, v_6)$ (see the left schedule of Fig. \ref{fig3tight}). The worst case ratio of the PAR algorithm cannot be less than $2$ as $\frac{C_{max}}{C_{max}^*}\rightarrow 2$ when $\epsilon \rightarrow 0$ for this instance.

\begin{figure}[ht]
  \centering
  \includegraphics[width=4.5in]{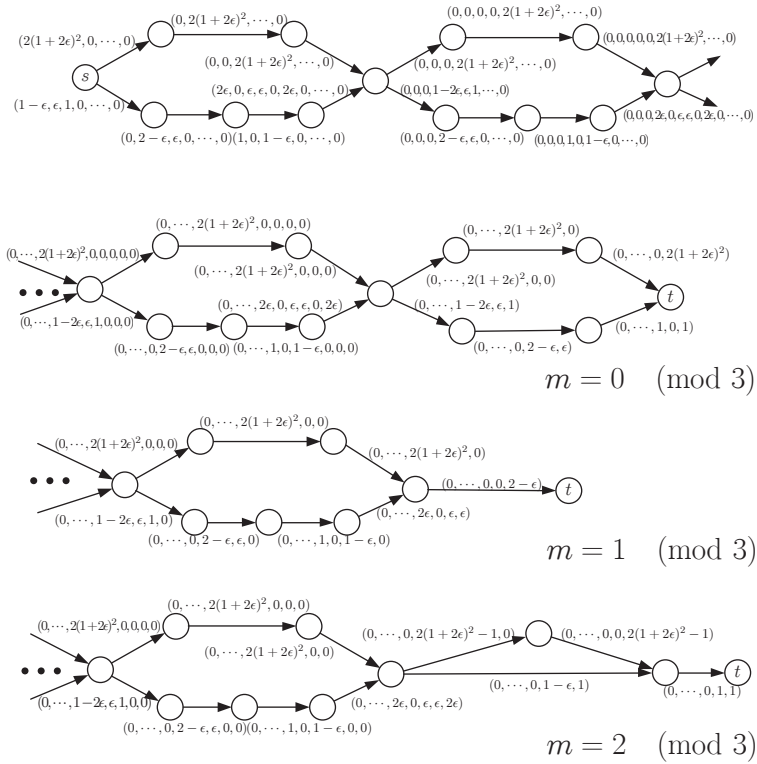}\\
  \caption{Example for fixed $m$}\label{figmtight}
\end{figure}

By extending and modifying the above examples to the general case, the instance described in Fig. \ref{figmtight} can be used to show that the performance ratio of the PAR algorithm cannot be less than $\rho$. $\Box$
\end{proof}

\section{Conclusions}\label{sec_end}
This paper has studied a combination problem of flow shop scheduling and the shortest path problem.  We show the hardness of this problem, and present some approximation algorithms. For future research, it would be interesting to find an approximation algorithm with a better performance ratio for this problem. The question whether $F2|\mathrm{shortest}~\mathrm{path}|C_{max}$ is NP-hard in the strong sense is still open. One can also consider the combination of other combinatorial optimization problems. All these questions deserve further investigation.

\section*{Acknowledgments}
This work has been supported by the Bilateral Scientific Cooperation Project  BIL10/10 between Tsinghua University and KU Leuven.


\begin{thebibliography}{}
\bibitem[Ahuja et al.(1993)]{AMO93}
Ahuja RK, Magnanti TL, Orlin JB (1993) Network Flows: Theory, Algorithms,
  and Applications. Prentice Hall, New Jersey

\bibitem[Aissi, Bazgan and Vanderpooten(2006)]{ABV06}
Aissi H, Bazgan C, Vanderpooten D (2006) Approximating min-max (regret)
  versions of some polynomial problems. In: Chen, D., Pardolos, P.M. (eds.)
  COCOON 2006, LNCS, vol 4112, pp 428--438. Springer, Heidelberg

\bibitem[Batagelj et al.(2000)]{Batagelj2000}
Batagelj V, Brandenburg FJ, Mendez P, Sen A (2000) The generalized shortest
  path problem. CiteSeer Archives

\bibitem[Bodlaender et al.(1994)]{Bodlaender1994}
Bodlaender HL, Jansen K,  Woeginger GJ (1994) Scheduling with incompatible
jobs. Disc Appl Math, 55: 219--232

\bibitem[Chen et al.(1996)]{Chen1996}
Chen B, Glass CA, Potts CN, Strusevich VA (1996) A new heuristic for
  three-machine flow shop scheduling. Oper Res 44: 891--898

\bibitem[Conway et al.(1971)]{Conway1971}
Conway RW, Maxwell W, Miller L (1967) Theory of scheduling. Reading

\bibitem[Dijkstra(1959)]{DIJ59}
Dijkstra EW (1959) A note on two problems in connexion with graphs. Numer Math  1:  269--271

\bibitem[Garey et al.(1976)]{GJS1976}
Garey MR, Johnson DS, Sethi R (1976) The complexity of flowshop and jobshop
  scheduling. Math Oper Res  1:  117--129

\bibitem[Garey and Johnson(1979)]{GJ79}
Garey MR, Johnson DS (1979) Computers and Intractability: A Guide to the Theory
  of NP-completeness. Freeman, San Francisco

\bibitem[Gonzalez and Sahni(1978)]{Gonzalez1978}
Gonzalez T, Sahni S (1978) Flowshop and jobshop schedules: complexity and
  approximation. Oper Res  26:  36--52

\bibitem[Hall(1988)]{Hall1998}
Hall LA (1998) Approximability of flow shop scheduling. Math Program
  82: 175--190

\bibitem[Johnson(1954)]{Joh54}
Johnson SM (1954) Optimal two- and three-stage production schedules with setup
  times included. Nav Res Logist Q 1:  61--68

\bibitem[Kouvelis and Yu(1997)]{KY97}
Kouvelis P, Yu G (1997) Robust discrete optimization and its applications.
  Kluwer Academic Publishers, Boston

\bibitem[R{\"o}ck and Schmidt(1982)]{Rock1982}
R{\"o}ck H, Schmidt G (1982) Machine aggregation heuristics in shop scheduling.
  Method Oper Res 45: 303--314

\bibitem[Wang and Cui(2012)]{WC12}
Wang Z, Cui Z (2012) Combination of parallel machine scheduling and vertex cover.
  Theor Comput Sci  460:  10--15

\bibitem[Wang et al.(2013)]{WHH13}
Wang Z, Hong W, He D (2013): Combination of parallel machine scheduling and
  covering problem. Working paper, Tsinghua University

\bibitem[Warburton(1987)]{Warburton87}
Warburton A (1987) Approximation of pareto optima in multiple-objective,
  shortest-path problems. Oper Res  35:  70--79
\end{thebibliography}
\end{document}